\documentclass[12pt,reqno]{amsart}
\usepackage{amsaddr}
\usepackage{graphicx}
\usepackage{setspace}
\usepackage{empheq}
\usepackage{cite}
\usepackage{textcomp}
\usepackage{dutchcal}
\usepackage{amssymb}
\usepackage{hyperref}
\usepackage{enumerate}

\usepackage{color}
\definecolor{MyLinkColor}{rgb}{0,0,0.4}

\newcommand{\R}{{\mathbb R}}

\newcommand{\ov}{\overline}

\newcommand{\ee}{\mathrm e}

\newcommand{\dx}{\mathrm d}
\newtheorem{thm}{Theorem}[section]
\newtheorem{prop}[thm]{Proposition}

\numberwithin{equation}{section}   

\title[Exact solutions for the ACC]{Explicit and exact solutions concerning the Antarctic Circumpolar Current with variable density in spherical coordinates}
\author[C.~I.~Martin]{Calin Iulian Martin$^*$}
\thanks{$^*$Corresponding author} 
\address{Faculty of Mathematics, University of Vienna, Oskar Morgenstern Platz 1, 1090 Vienna, Austria \\
	E-mail address: \href{mailto:calin.martin@univie.ac.at}{\textsf{calin.martin@univie.ac.at}}}

\author[R.~Quirchmayr]{Ronald Quirchmayr}
\address{Department of Mathematics, KTH Royal Institute of Technology, Lindstedtsv\"agen 25, 100 44 Stockholm, Sweden \\
	E-mail address: \href{mailto:ronaldq@kth.se}{\textsf{ronaldq@kth.se}}}

\begin{document}

\maketitle	

\begin{abstract}
\noindent
We use spherical coordinates to devise a new exact solution to the governing equations of geophysical fluid dynamics for an inviscid and incompressible fluid with a general 
density distribution and subjected to forcing terms.
The latter are of paramount importance for the modeling of realistic flows-that is, flows that are observed in some averaged sense in the ocean.
Owing to the employment of spherical coordinates we do not need to resort to approximations (e.g. of $f$- and $\beta$-plane type) that simplify the geometry in the governing equations.
Our explicit solution represents a steady purely-azimuthal stratified flow with a free surface, that---thanks to the inclusion of forcing terms and the consideration of the Earth's 
geometry via spherical coordinates---makes it  suitable for describing the Antarctic Circumpolar Current and enables an in-depth analysis of the structure of this flow.
In line with the latter aspect, we employ functional analytical techniques to prove that the free surface distortion is defined in a unique and implicit way by means 
of the pressure applied at the free surface. We conclude our discussion by setting out relations between the monotonicity of the surface pressure and the monotonicity of 
the surface distortion that concur with the physical expectations.

\vspace{1em}
\end{abstract}

\section{Introduction}
\noindent
To a large extent, the analytical studies concerning the fluid flows were (are) concentrated on well-posedness or regularity type results. Our contribution here aligns with a recent
tendency---initiated by Constantin \cite{CGeoResLett, CoGeoPhys, CoPhysOc13, CoPhysOc14} and Constantin and Johnson---\cite{CJ, CJaz, CJazAcc, CJPoF} of devising explicit and 
exact solutions to the governing equations of geophysical fluid dynamics (GFD) that describe surface waves and their interactions with the underlying currents which are ubiquitous 
in Earth's ocean basins. A selection of further works in this direction is \cite{ChuEsc, ChuIonYang, HenEjmb, HsuMarNA, Ion, MarNonl, AMJPhA, AMApplAna, MatMatJnmp}.

More precisely we are concerned here with providing an analytical solution to the governing equations for water flows (with their boundary conditions) which represents
a flow that moves completely around the Earth on a circular path. This
solution describes an incompressible inviscid stratified steady flow moving purely in the azimuthal direction, 
i.e.~the velocity profile and the pressure is described below and up to the free surface as a function of depth and the angle of latitude.
As such, this solution
is suitable for the depiction of the Antarctic Circumpolar Current (ACC), without doubt the most significant current in Earth's oceans. Indeed, ACC is the only current that fully surrounds
the polar axis. Its massiveness is reflected by the great area it occupies---flowing eastwards through the southern regions of Atlantic, Indian and Pacific Oceans 
along $23.000$ km, extending in places over $2000$ km in width, cf. \cite{Fir, Ivch, Olb, Rin}
---and as such, by the huge volumes of water it transports estimated to be between 165 million and 182 million cubic meters of water every second, cf. \cite{dtwcc}, which represents more 
than 100 times the flow of all the rivers on Earth. Recent mathematical studies investigating ACC's properties can be found in 
\cite{CJazAcc, HazMary, HazDCDS, HsuMarAcc, Mary, Quir}. 

The new aspect of our investigation---compared to the existing mathematical
literature on ACC---is the presence of density stratification. This accommodates observed 
sharp changes in water density known as fronts, cf. \cite{Phillips}. The two main fronts of the ACC are the Subantarctic Front to the north and the Polar Front further south, 
see Figure \ref{fig: fronts}.
\begin{figure}[h]
\begin{center}
\includegraphics*[width=0.75\textwidth]{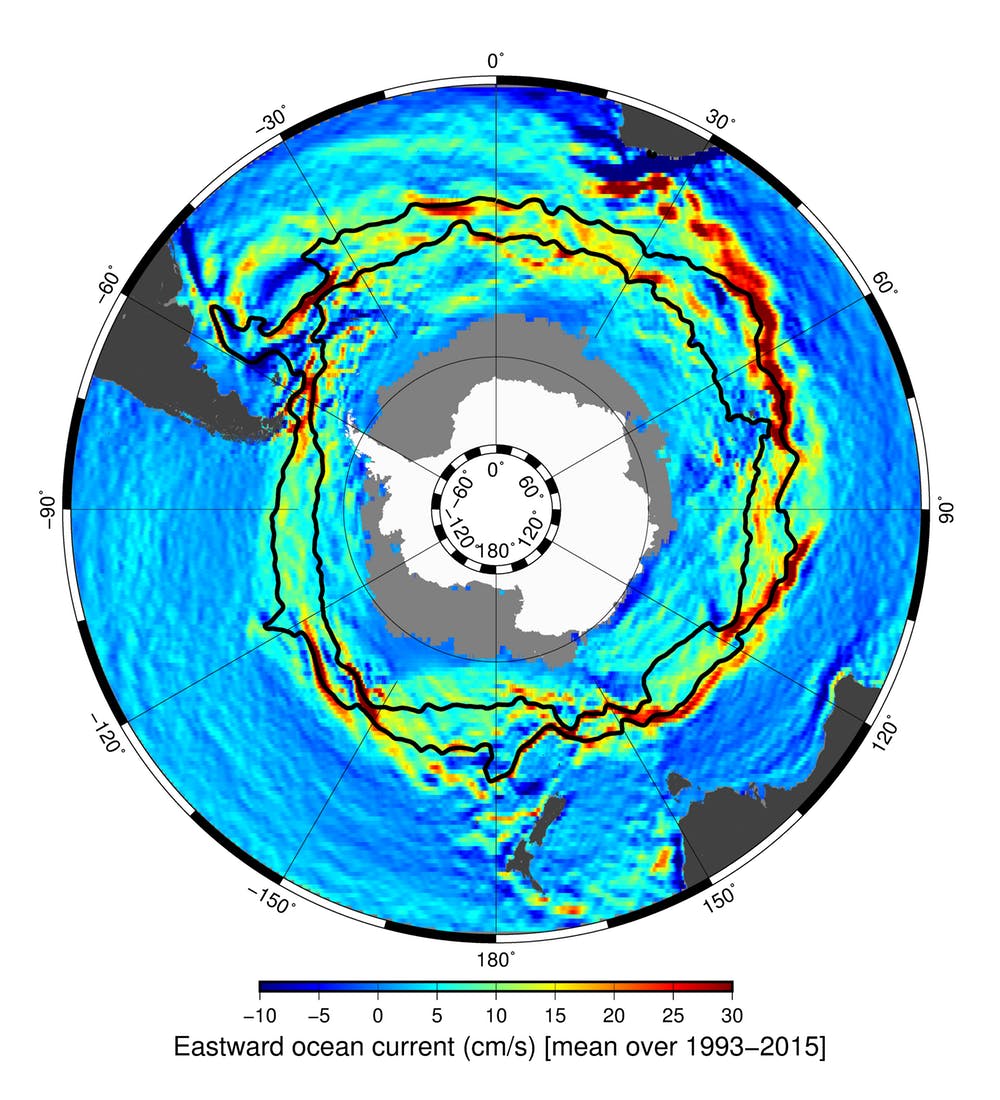}
\end{center}
\caption{Variations in density and the mean eastward current speed. The black lines in the image represent the fronts of ACC. Image credit: Hellen Phillips (Senior Research Fellow, Institute for Marine and Antarctic Studies, 
University of Tasmania), Benoit Legresy (CSIRO) and Nathan Bindoff (Professor of Physical Oceanography, Institute for Marine and Antarctic Studies, University of Tasmania)} 
\label{fig: fronts}
\end{figure}
We remark that the recent study \cite{Kob} clearly indicates the role of a thorough knowledge of the impact of enhanced stratification in the Southern Ocean on the mechanisms behind the glacial-interglacial
ocean carbon cycle variations.

Apart from the two aspects pertaining to density variations in ACC we would like to remark that stratification is of significant importance in the understanding of GFD in general. Indeed, 
large scale oceanic processes exhibit and experience pronounced density variations, most commonly due to fluctuations in the fluid temperature or salinity,  
cf. \cite{CGeoResLett,CJ,CJPoF,FedBr,KesMcP,McC}. Therefore, it is of high significance to achieve a comprehensive analytical understanding of stratified flows which could potentially
serve more applied endeavors. However, allowing for variable density in water flows greatly complicates the analysis of an already challenging problem. Indeed, we note that even in 
the scenario of two-dimensional water flows the literature manifests a pronounced scarcity: recent steps ahead being
undertaken in \cite{BasDCDS, CWW, CI, CJ, cim, CICpam, EMM, GeyQui, HM1, HM2, HenMarJDE, HenMarARMA, Wal1, Whe}.

Exact and/or explicit solutions describing geophysical flows constitute a special and very rare event, reason being the intricacy of the involved mathematical problem. 
Once they are available they secure ways to produce more physically realistic and observed flows, by means of asymptotic \cite{CJJPhysOc19}, or multiple scale methods \cite{CJ}.
Faithfull to this realization,
we derive here explicit solutions for the pressure and for the velocity field beneath the free surface: the latter is itself given in an implicit form that relates it 
to the pressure at the surface. Building on the approaches by Constantin and Johnson \cite{CJaz, CJazAcc}, we use spherical coordinates to devise purely-azimuthal, 
depth-dependent varying flows, that verify the GFD equations and their boundary conditions. 
Moreover, our solutions exhibit a general stratification (that varies with depth and latitude) and forcing terms believed to be responsible for the dynamical balance 
of the ACC, cf. \cite{CJazAcc, MarshForcing}. Finally we point to \cite{Marsh16, Danabas} for the relevance of stratification to maintain the equilibrium of ACC:
baroclinic instability (arising from stratification) generates eddy-induced cells (acting to flatten the isopycnals) that counterbalance the wind-driven Ekman cell 
(acting to steepen isopycnals).

The organization of the paper is the following: After introducing in Section \ref{presentation} the governing equations (in spherical coordinates) and their boundary conditions for geophysical flows, we 
derive in Section \ref{explicit_sol} explicit solutions that render the velocity in the azimuthal direction and the corresponding pressure function. The implicit (exact) solution that 
describes the free surface is achieved in Section \ref{implicit_sol} by means of a functional analytic argument. Relations between the monotonicity of the free surface and the surface pressure are presented in the last section of the paper.

\section{Physical problem and governing equations}\label{presentation}
\noindent
In this section we provide the governing equations for geophysical flows written in spherical coordinates to accommodate the shape of the Earth, together with the boundary conditions for the free surface and a rigid bed.

We will work in a system of right handed coordinates $(r, \theta, \varphi)$ where $r$ denotes the distance to the centre of the sphere, $\theta\in [0,\pi]$  is the polar angle
(the convention being that $\pi/2-\theta$ is the angle of latitude) $\varphi\in [0, 2\pi]$ is the azimuthal angle (the angle of longitude). While in this coordinate system
the North and South poles are located at $\theta=0,\pi$, respectively, the Equator sits on $\theta=\pi/2$, the Antarctic Circumpolar Current is situated at $\theta=3\pi/4$.
The unit vectors in this system are 
\begin{align*}
{\bf e}_r&=(\sin\theta\cos\varphi,\sin\theta\sin\varphi, \cos\theta),\\
{\bf e}_{\theta}&=(\cos\theta\cos\varphi,\cos\theta\sin\varphi, -\sin\theta),\\
{\bf e}_{\varphi}&=(-\sin\varphi,\cos\varphi,0),
\end{align*}
with $\bf{e}_{\varphi}$ pointing from West to East and $\bf{e}_{\theta}$ from North to South, cf. Figure \ref{fig: sphericalsyst}.
\begin{figure}[h]
\begin{center}
\includegraphics*[width=0.75\textwidth]{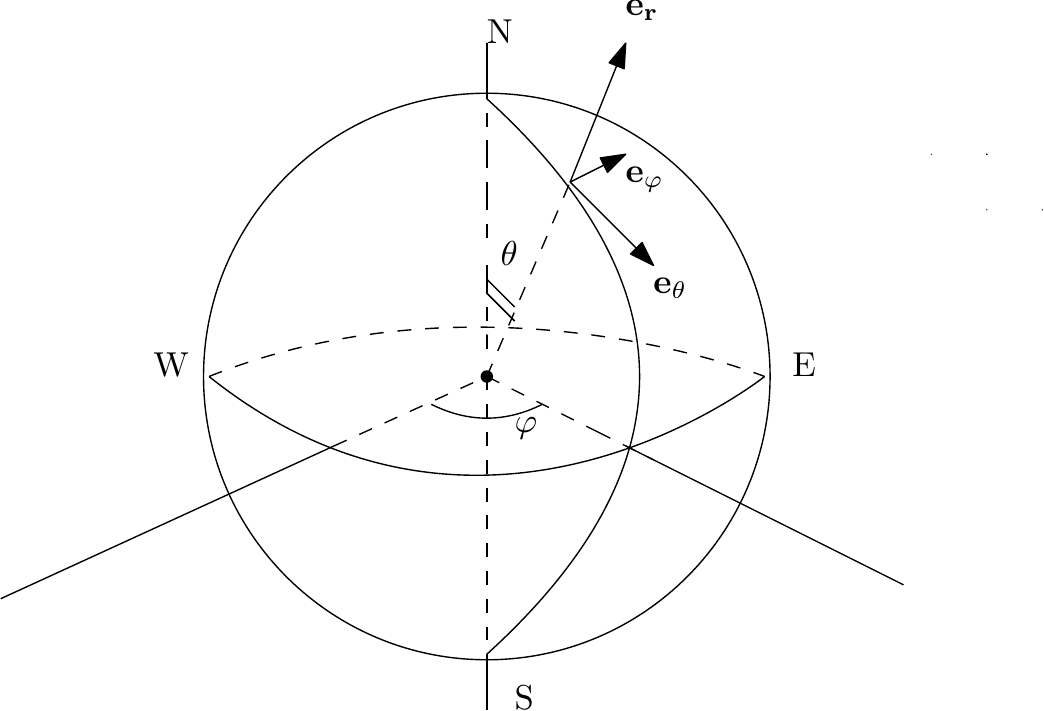}
\end{center}
\caption{The spherical coordinate system: $\theta$ is the polar angle, $\varphi$ is the azimuthal angle (the angle of longitude) and $r$ represents the distance to the origin.} 
\label{fig: sphericalsyst}
\end{figure}

Throughout this paper we make the following simplifying assumption on the location of the Antarctic Circumpolar Current. We assume that the angle of latitude $\theta$ satisfies
\begin{equation} \label{latitude_ACC}
\frac{3\pi}{4} - \frac{\pi}{18}\leq\theta \leq \frac{3\pi}{4} + \frac{\pi}{18}.
\end{equation}

We are guided in our study by the observations made in Maslowe \cite{Mas} asserting that the Reynolds number is, in general, extremly large for this oceanic flows. 
Accordingly we will consider incompressible and inviscid flows. However, our analysis incorporates a general density distribution $\rho=\rho(r,\theta)$.

Denoting with 
\begin{equation*}
{\bf u}=u{\bf e}_r+v{\bf e}_{\theta}+w{\bf e}_{\varphi} 
\end{equation*}
the velocity field we have that, cf. \cite{CJaz}, the governing equations in the $(r, \theta, \varphi)$ 
coordinate system are the Euler's equations,
 \begin{equation}\label{Eulereq}
 \begin{array}{rcl}
  u_t+uu_r +\frac{v}{r}u_{\theta}+\frac{w}{r\sin\theta}u_{\varphi}-\frac{1}{r}(v^2+w^2) & =& -\frac{1}{\rho}p_r +F_r\\
  & &\\
    v_t+uv_r +\frac{v}{r}v_{\theta}+\frac{w}{r\sin\theta}v_{\varphi}+\frac{1}{r}(uv-w^2\cos\theta) & =& -\frac{1}{\rho}\frac{1}{r}p_{\theta} +F_{\theta}\\
    & &\\
    w_t+uw_r +\frac{v}{r}w_{\theta}+\frac{w}{r\sin\theta}w_{\varphi}+\frac{1}{r}(uw+vw\cot\theta)&=& -\frac{1}{\rho}\frac{1}{r\sin\theta}p_{\varphi} +F_{\varphi},
  \end{array}
\end{equation}
(where $p(r,\theta,\varphi)$ is denotes the pressure in the fluid and $(F_r,F_{\theta}, F_{\varphi})$ is the body-force vector) and the equation of mass conservation
\begin{equation}\label{masscons}
 \frac{1}{r^2}\frac{\partial}{\partial r}(\rho r^2 u)+\frac{1}{r\sin\theta}\frac{\partial}{\partial\theta}(\rho v\sin\theta)+\frac{1}{r\sin\theta}\frac{ \partial (\rho w)}{\partial  \varphi}=0.
\end{equation}

To capture effects of Earth's rotation we associate $({\bf e}_r, {\bf e}_{\theta}, {\bf e}_{\varphi})$ with a fixed point on the sphere which is rotating about its polar axis.
Consequently, we consider on the left of \eqref{Eulereq} the Coriolis term 
\begin{equation*}
2{\bf \Omega} \times {\bf u},
\end{equation*}
where 
\begin{equation*}
{\bf \Omega} =\varOmega({\bf e}_r\cos\theta-{\bf e}_{\theta}\sin\theta),
\end{equation*}
(with $\varOmega\approx 7.29\times 10^{-5}$ rad $s^{-1}$ being the constant rotation speed of Earth) and the centripetal acceleration
\begin{equation*}
{\bf \Omega}\times ({\bf \Omega}\times {\bf r}),
\end{equation*}
where ${\bf r}=r{\bf e}_r$. The latter two quantities add up to
\begin{align*}
2&\varOmega\big[-w(\sin\theta){\bf e}_r -w(\cos\theta){\bf e}_{\theta}+(u\sin\theta+v\cos\theta){\bf e}_{\varphi}\big] \\
&-r\varOmega^2\big[(\sin^2\theta) {\bf e}_r+(\sin\theta\cos\theta){\bf e}_{\theta}\big]
\end{align*}
with respect to the $({\bf e}_r, {\bf e}_{\theta}, {\bf e}_{\varphi})$ basis.

Concerning the body-force vector it holds that (cf.~\cite{CJazAcc}) the gravitational constant acts in $r$-direction, a general body-force $G(r,\theta)$ acts in $\theta$-direction, no force is acting in $\varphi$-direction. That is, the body-force vector is 
\begin{equation*}
(-g,G(r,\theta), 0).
\end{equation*}

The equations of motion are filled up by the boundary conditions as follows.
On the free surface $r=R+h(\theta,\varphi)$ we have the dynamic boundary condition
\begin{equation}\label{surfacepressure}
 p=P(\theta, \varphi),
\end{equation}
and the kinematic boundary condition 
\begin{equation}\label{kin_surf}
u=\frac{v}{r}\frac{\partial h}{\partial\theta}+\frac{w}{r\sin\theta}\frac{\partial h}{\partial\varphi}.
\end{equation}
We ask that the bottom of the ocean, described by $r=d(\theta,\varphi)$, is an impenetrable boundary by demanding 
\begin{equation}\label{kin_bed}
 u=\frac{v}{r}\frac{\partial d}{\partial\theta}+\frac{w}{r\sin\theta}\frac{\partial d}{\partial\varphi}
\end{equation}
to hold there.

The solutions that we are seeking for, represent a steady flow that propagates only in the azimuthal direction with no variations in this direction. Thus, the velocity field is characterized by 
$u=v=0$ and $w=w(r,\theta)$. Moreover, $p=p(r,\theta)$, $h=h(\theta)$, $d=d(\theta)$.
Consequently, while the two kinematic boundary conditions \eqref{kin_surf} and \eqref{kin_bed} and the equation of mass conservation \eqref{masscons}
are trivially satisfied, the Euler equations governing such special 
flows read
\begin{equation}\label{specialflow}
 \left\{\begin{array}{rcl}
         -\frac{w^2}{r}-2\varOmega w\sin\theta-r\varOmega^2\sin^2 \theta & =& -\frac{1}{\rho}p_r-g,\\
         & &\\
          -\frac{w^2}{r}\cot\theta- 2\varOmega w\cos\theta -r\varOmega^2\sin\theta\cos\theta & = & -\frac{1}{\rho r}p_{\theta}+G(r,\theta),\\
          &  &\\
0& =& -\frac{1}{\rho}\frac{1}{r\sin\theta}p_{\varphi}.
          \end{array}\right.
\end{equation}

\section{Exact explicit and implicit solutions}\label{solutions}
\noindent
\subsection{Explicit solutions for azimuthal velocity and pressure}\label{explicit_sol}
The first outcomes of this section concern the velocity $w$ and the pressure $p$. These two quantities can be determined by means of explicit formulas. To derive the latter we 
use the third equation in \eqref{specialflow} to infer that $p$ is independent of $\varphi$, i.e.~$p=p(r,\theta)$. 
System \eqref{specialflow} can now be simplified as 
\begin{align}
\rho\frac{(w+\Omega r\sin\theta)^2}{r}&=p_r+g\rho  \\
\rho r\cot\theta \frac{(w+\Omega r\sin\theta)^2}{r}&=p_{\theta}-\rho r G(r,\theta).
\end{align}
Let us denote
\begin{equation} \label{Z}
Z = Z(r,\theta) \coloneqq \frac{(w+\Omega r\sin\theta)^2}{r}
\end{equation}
to rewrite the above system as
\begin{align}
\rho Z&=p_r+g\rho \label{system_Z_1} \\
(\rho r\cot\theta) Z&=p_{\theta}-\rho r G(r,\theta) \label{system_Z_2}.
\end{align}
The elimination of the pressure $p$ from the system \eqref{system_Z_1}--\eqref{system_Z_2} leads to 
\begin{equation*}
\left(\rho  Z\right)_{\theta}-\left(\rho r\cot\theta  Z\right)_r =g\rho_{\theta}+\frac{\partial}{\partial r}\big(\rho r G(r,\theta)\big)
\end{equation*}
which can be rewritten as 
\begin{equation} \label{eq_U'}
-r\cos\theta\left(\rho r Z\right)_r +\sin\theta \left(\rho r Z\right)_{\theta}=  (r \sin\theta) \left[g\rho_{\theta}(r,\theta)+
\frac{\partial}{\partial r}\big(\rho r G(r,\theta)\big)\right].
\end{equation}
We further set 
\begin{equation} \label{U}
U \coloneqq \rho r Z
\end{equation}
and appeal to the method of characteristics to solve \eqref{eq_U'}; i.e., we set out to solve
\begin{equation}\label{eq_U}
-(r\cos\theta) U_r+(\sin\theta) U_{\theta}=   (r \sin\theta) \left[g\rho_{\theta}(r,\theta)+
\frac{\partial}{\partial r}\big(\rho r G(r,\theta)\big)\right]. 
\end{equation}
Along these lines we seek curves 
$s\rightarrow (r(s),\theta(s))$ satisfying 
\begin{gather}\label{charac_eq}
\dot{r}(s)=-r(s)\cos\theta(s), \qquad \dot\theta(s)=\sin\theta(s).
\end{gather}
Note that any $(r(s),\theta(s))$ that satisfies the previous system also obeys the equation
\begin{equation}\label{rthetaprop}
\frac{\dx}{\dx s}\Big(r(s)\sin\theta(s)\Big)=0\quad{\rm for}\,\,{\rm all}\quad s\in \mathbb{R}
\end{equation}
The choice \eqref{charac_eq} transforms equation \eqref{eq_U} into
\begin{align} \label{paramformU}
\begin{aligned}
\frac{\dx}{\dx s}\Big(&U\big(r(s),\theta(s)\big)\Big)= \\  
\big(&r(s) \sin\theta(s)\big) \left[g\rho_{\theta}\big(r(s),\theta(s)\big)+
\frac{\partial}{\partial r}\Big(\rho\big(r(s),\theta(s)\big) r(s) G\big(r(s),\theta(s)\big)\Big)\right],
\end{aligned}
\end{align}
whose resolution depends upon finding suitable solutions to the characteristic equations \eqref{charac_eq}.
 The integration of \eqref{charac_eq} yields the general solution  $(\tilde{r}(s),\tilde{\theta}(s))$ given by
\begin{gather}
\tilde{r}(s)=c_1 \ee^s+c_2 \ee^{-s},\quad c_1,c_2\in\mathbb{R}\\
\tilde{\theta}(s)=\arccos\left(\frac{1-c\ee^{2s}}{1+c \ee^{2s}}\right),\quad c\geq 0,
\end{gather}
where $c,c_1,c_2$ are constants of integration. To find $U$ as a function of $r$ and $\theta$ we proceed as follows:
for given $(r,\theta)$, we search for an $s_0$ such that 
\begin{equation}\label{initial_values}
\tilde{\theta}(s_0)=\theta, \quad\tilde{r}(s_0)=r.
\end{equation}
 We choose now the constant $c$ such that $\tilde{\theta}(0)=\frac{\pi}{2}$. This gives $c=1$, that is 
\begin{equation}
\tilde{\theta}(s)=\arccos\left(\frac{1-\ee^{2s}}{1+ \ee^{2s}}\right),
\end{equation}
From the property \eqref{rthetaprop} we see that 
\begin{equation*}
\tilde{r}(0) \sin\big(\tilde{\theta}(0)\big)=  \tilde{r}(s_0)\sin\big(\tilde{\theta}(s_0)\big)=  r\sin\theta,
\end{equation*} 
which, after using $\tilde{\theta}(0)=\frac{\pi}{2}$, gives
$\tilde{r}(0)=r\sin\theta$.

It is easy to see that
\begin{equation*}
s_0=\frac{1}{2}\ln\frac{1-\cos\theta}{1+\cos\theta}\eqqcolon f(\theta)
\end{equation*}
 is the unique 
element satisfying the first equation in \eqref{initial_values}. To determine $c_1$ and $c_2$ we solve
\begin{gather}
\tilde{r}(f(\theta))=r  \label{rftheta} \\  \tilde{r}(0)=r\sin\theta \label{rf0}
\end{gather}
that is equivalent to 
\begin{gather}
c_1+c_2+(c_2-c_1)\cos\theta=r\sin\theta \nonumber \\ c_1+c_2=r\sin\theta, \nonumber
\end{gather}
whose unique solution (for $\theta\neq \pi/2$) is 
\begin{equation*}
c_1=c_2=\frac{r\sin\theta}{2}.
\end{equation*}
Let us denote with
 $(\ov{r}(s),\ov{\theta}(s))$ the special characteristic solution to \eqref{charac_eq} that satisfies \eqref{initial_values}.
That is,
\begin{equation*}
\ov{r}(s)=\frac{r\sin\theta}{2}(\ee^s+\ee^{-s}),\quad \ov{\theta}(s)=\arccos\left(\frac{1-\ee^{2s}}{1+ \ee^{2s}}\right).
\end{equation*}
Inserting $(\ov{r}(s),\ov{\theta}(s))$ in the formula \eqref{paramformU},  integrating afterwards with respect to $s$ from 
$0$ to $f(\theta)$,  taking into account \eqref{rftheta}--\eqref{rf0}, and setting 
\begin{equation} \label{H}
H(r,\theta) \coloneqq r\rho(r,\theta)G(r,\theta),
\end{equation}
we obtain
\begin{equation} \label{U_formula}
U(r,\theta)-U(r\sin\theta,\pi/2)=r\sin\theta\int_0^{f(\theta)}\left[g\rho_{\theta}(\ov{r}(s),\ov{\theta}(s)) + H_r(\ov{r}(s),\ov{\theta}(s))    \right] \, \dx s,
\end{equation}
where we have used that $\frac{\dx}{\dx s}\big(\ov{r}(s)\sin\ov{\theta}(s)\big)=0$ for all $s$, which implies that 
\begin{equation*}
\ov{r}(s)\sin\left(\ov{\theta}(s)\right)=\ov{r}(0)\sin\left(\ov{\theta}(0)\right)=r\sin\theta\,\,{\rm for}\,\,{\rm all}\,\, s.
\end{equation*}
From the definitions of the functions $Z$ and $U$, cf.~\eqref{Z} and \eqref{U}, respectively, we finally obtain that the azimuthal velocity $w$ is given by 
the formula
\begin{align}\label{azim_vel}
\begin{aligned}
w(r,\theta)=&-\Omega r\sin\theta  \\
&+\sqrt{ \frac{ F(r\sin\theta)+ r\sin\theta\int_0^{f(\theta)}\left[g\rho_{\theta}(\ov{r}(s),\ov{\theta}(s)) +H_r(\ov{r}(s),\ov{\theta}(s))    \right] \, \dx s    }{\rho(r,\theta) }    },
\end{aligned}
\end{align}
where $t \mapsto F(t)$ denotes a given arbitrary smooth function.

To determine the pressure $p$ we first infer from \eqref{Z}, \eqref{system_Z_1},  \eqref{system_Z_2}, \eqref{U}, \eqref{H} and \eqref{U_formula} that
\begin{align}
p_r(r,\theta) &= -g \rho(r,\theta) +\frac{1}{r} \bigg(
F( r \sin \theta) + r \sin \theta \int^{f(\theta)}_0 [g\rho_\theta (\ov r,\ov \theta) + H_r(\ov r,\ov \theta)  ] \,  \dx s \bigg), 
\label{p_r} \\
p_\theta(r,\theta) &= H(r,\theta) + \cot \theta \,
\bigg[ F( r \sin \theta) + r \sin \theta \int^{f(\theta)}_0 [g\rho_\theta (\ov r,\ov \theta) + H_r(\ov r,\ov \theta)  ] \,  \dx s  \bigg]. \label{p_theta}
\end{align}
Integrating both sides in \eqref{p_r} with respect to $r$ yields that 
\begin{equation} \label{p_1}
p(r,\theta) = C(\theta) - g \int^r_a \rho(\tilde r,\theta) \,\dx \tilde r 
+ \int^{r \sin \theta}_{a  \sin \theta} \bigg[\frac{F(y)}{y} +  \mathcal F(y,\theta) \bigg] \,\dx y,
\end{equation} 
where $a\in\R$ is a constant, the map $\theta\mapsto C(\theta)$ will be established in the sequel, and $\mathcal F$ is given by
\begin{equation} \label{mathcal_F}
\mathcal F(y,\theta) \coloneqq \int^{f(\theta)}_0 \bigg[ g \rho_\theta \bigg(y \, \frac{\ee^s+\ee^{-s}}{2}, \ov\theta(s) \bigg) + H_r \bigg(y \, \frac{\ee^s+\ee^{-s}}{2}, \ov\theta(s)   \bigg)  \bigg]   \, \dx s.
\end{equation}
To determine $C$, we differentiate \eqref{p_1} with respect to $\theta$, which, by an application of the chain rule, yields that
\begin{align}
	\begin{aligned} \label{p_theta_2}
p_\theta(r,\theta) = 
C'(\theta) 
+   H(r,\theta) - H(a,\theta) 
+ \cot \theta \, \big[ F(r  \sin \theta) -F(a  \sin \theta)\big] .
\end{aligned}
\end{align}
Here we have exploited that for every fixed $y$,
\begin{align*}
 \mathcal F_\theta (y,\theta) &= f'(\theta)
  \bigg[ g \rho_\theta \bigg(y \, \frac{\ee^s+\ee^{-s}}{2}, \ov\theta(s) \bigg) + H_r \bigg(y \, \frac{\ee^s+\ee^{-s}}{2}, \ov\theta(s)   \bigg)  \bigg] \Bigg|_{s=f(\theta)}  \\
  &= \frac{1}{\sin \theta} \bigg[  g \rho_\theta \bigg( \frac{y}{\sin\theta}, \theta(s) \bigg) + H_r \bigg(\frac{y}{\sin\theta}, \theta(s)  \bigg) \bigg],
\end{align*}
which enabled us to write
\begin{equation*} 
 \int^{r \sin \theta}_{a  \sin \theta}   \mathcal F_\theta (y,\theta)  \,\dx y = 
 \int^r_a   \big[g \rho_\theta(\tilde r,\theta) + H_r(\tilde r,\theta) \big]  \, \dx \tilde r.
\end{equation*}
A comparison of  \eqref{p_theta_2} and \eqref{p_theta} allows us to deduce that
\begin{equation} \label{C'}
C'(\theta) = F (a \sin \theta) \cot \theta + \mathcal F (a \sin \theta,\theta) \, a \cos \theta + H(a,\theta).
\end{equation}
We conclude from \eqref{p_1} and \eqref{C'} that the pressure satisfies
\begin{align} \label{p_formula}
	\begin{aligned}
p(r,\theta) = b &- g \int^r_a \rho(\tilde r,\theta) \,\dx \tilde r 
+ \int^{r \sin \theta}_{a  \sin \theta} \bigg[\frac{F(y)}{y} +  \mathcal F(y,\theta) \bigg] \,\dx y \\
&+ \int^\theta_{3\pi/4} \big[ F (a \sin \tilde \theta) \cot \tilde \theta 
+ \mathcal F (a \sin \tilde \theta,\tilde\theta) \, a \cos \tilde \theta + H(a, \tilde \theta) \big] \, \dx \tilde \theta,
\end{aligned}
\end{align}
where $a,b\in\R$ are constants, $F$ is an arbitrary smooth function, and $\mathcal F$ is given according to \eqref{mathcal_F}.

\subsection{Implicit exact solution for the free surface}\label{implicit_sol}
By an application of the implicit function theorem we will show the existence of a unique function $h=h(\theta)$ representing disturbances of the flat water surface 
caused by deviations of the pressure distribution from the pressure required to maintain a free surface following the curvature of the Earth.

According to the dynamic boundary condition at the surface \eqref{surfacepressure}
and our assumption of a purely azimuthal flow that may only vary in $\theta$, the pressure $p$ given by \eqref{p_formula} and evaluated at the free surface $r=R+ h(\theta)$, 
denoted by $P(\theta)$, satisfies
\begin{align} \label{P_formula}
\begin{aligned}
P(\theta) = b &- g \int^{R+ h(\theta)}_a \rho( r,\theta) \,\dx  r 
+ \int^{[R+ h(\theta)] \sin \theta}_{a  \sin \theta} \bigg[\frac{F(y)}{y} +  \mathcal F(y,\theta) \bigg] \,\dx y \\
&+ \int^\theta_{3\pi/4} \big[ F (a \sin \tilde \theta) \cot \tilde \theta 
+ \mathcal F (a \sin \tilde \theta,\tilde\theta) \, a \cos \tilde \theta + H(a, \tilde \theta) \big] \, \dx \tilde \theta.
\end{aligned}
\end{align}
Formula \eqref{P_formula} relates the pressure $P(\theta)$ acting onto the ocean surface to deformations of the free surface described by $h(\theta)$. To ensure comparability of the involved variables, we bring equation \eqref{P_formula} in dimensionless form. For this purpose we set $h\equiv 0$ in \eqref{P_formula}. This corresponds to the undisturbed ocean surface following the Earth's curvature. The pressure $P_0$ that corresponds to the flat surface shape is given by
\begin{align} \label{P_0_formula}
\begin{aligned}
P_0(\theta) = b &- g \int^{R}_a \rho( r,\theta) \,\dx  r 
+ \int^{R \sin \theta}_{a  \sin \theta} \bigg[\frac{F(y)}{y} +  \mathcal F(y,\theta) \bigg] \,\dx y \\
&+ \int^\theta_{3\pi/4} \big[ F (a \sin \tilde \theta) \cot \tilde \theta 
+ \mathcal F (a \sin \tilde \theta,\tilde\theta) \, a \cos \tilde \theta + H(a, \tilde \theta) \big] \, \dx \tilde \theta.
\end{aligned}
\end{align}
Let us denote by $P_a$ the atmospheric pressure at the ocean surface, which we assume to be constant along the entire range of the ACC. This is a reasonable simplification which, in combination with \eqref{P_0_formula}, implies that the constant $P_a$ can be written as
\begin{equation} \label{P_a_formula}
P_a \equiv b - g \int^{R}_a \rho( r, 3\pi/4) \,\dx  r 
+ \int^{\frac{R}{\sqrt{2}}}_{\frac{a}{\sqrt{2}}} \bigg[\frac{F(y)}{y} +  \mathcal F(y,3\pi/4) \bigg] \,\dx y .
\end{equation}
Dividing \eqref{P_formula} by $P_a$ yields that the dimensionless pressure $\mathcal P \coloneqq P/P_a$ satisfies
\begin{align} \label{P_nondim_formula}
\begin{aligned}
0=&-\mathcal{P} (\theta)+\frac{b}{P_a}-\frac{g}{P_a}\int_a^{[1+\mathcal{h}(\theta)]R} \rho(r,\theta) \, \dx r
+\frac{1}{P_a} \int_{a\sin\theta}^{[1+\mathcal{h}(\theta)]R\sin\theta}\left[\frac{F(y)}{y}+\mathcal{F}(y,\theta)\right] \, \dx y \\
&+\frac{1}{P_a}\int_{3\pi/4}^{\theta} [F(a\sin\tilde{\theta})\cot\tilde{\theta}   +\mathcal{F}(a\sin\tilde{\theta},\tilde\theta)   a\cos\tilde{\theta} +H(a,\tilde{\theta})] \, \dx \tilde{\theta} 
\eqqcolon \mathfrak{F}(\mathcal{h},\mathcal{P}),
\end{aligned}
\end{align}
where the dimensionless $\mathcal h (\theta)\coloneqq \frac{h(\theta)}{R}$ measures the deviation from the flat ocean surface at $r=R$. 

According to our assumption \eqref{latitude_ACC} we restrict the range of $\theta$ to the compact interval
\begin{equation*}
I\coloneqq \bigg[\frac{3\pi}{4} - \frac{\pi}{18},\frac{3\pi}{4}+ \frac{\pi}{18} \bigg]
\end{equation*}
and denote by $\mathcal C (I)$ the space of continuous functions $f\colon I \to \R$, equipped with the supremum norm 
\begin{equation}
\|f\| = \sup_{\theta\in I} |f(\theta)|.
\end{equation}
Furthermore we denote by $B$ the open ball of radius $10^{-5}$ centered at the origin in $\mathcal C (I)$; we note that any realistic disturbance $\mathcal h$ is contained in $B$.
We write now equation \eqref{P_nondim_formula} in the operatorial form
\begin{equation}
\mathfrak{F} \colon B \times \mathcal C (I) \to \mathcal C (I),
\end{equation}
where $\mathfrak{F}(\mathcal{h},\mathcal{P})$ is the right-hand side of \eqref{P_nondim_formula}.
We aim to implicitly resolve the dimensionless surface disturbance $\mathcal h$ as a function of the dimensionless surface pressure $\mathcal P$ by means of the operator equation 
\begin{equation} \label{operator_eq}
\mathfrak{F} (\mathcal h,\mathcal P) = 0.
\end{equation}
We note first that \eqref{operator_eq} is clearly satisfied for $\mathcal h \equiv 0$ and the corresponding dimensionless pressure $\mathcal P_0$, which is obtained by dividing both sides in \eqref{P_0_formula} by $P_a$:
\begin{align} \label{P_0__nondim_formula}
\begin{aligned}
\mathcal P_0(\theta) = \frac{b}{P_a} &- \frac{g}{P_a} \int^{R}_a \rho(\tilde r,\theta) \,\dx \tilde r 
+ \frac{1}{P_a} \int^{R \sin \theta}_{a  \sin \theta} \bigg[\frac{F(y)}{y} +  \mathcal F(y,\theta) \bigg] \,\dx y \\
&+ \frac{1}{P_a} \int^\theta_{3\pi/4} \big[ F (a \sin \tilde \theta) \cot \tilde \theta 
+ \mathcal F (a \sin \tilde \theta,\tilde\theta) \, a \cos \tilde \theta + H(a, \tilde \theta) \big] \, \dx \tilde \theta.
\end{aligned}
\end{align}
That is, we have that 
\begin{equation}\label{triv_sol}
 \mathfrak{F}(0,\mathcal{P}_0)=0.
\end{equation}
Having obtained a trivial solution of \eqref{operator_eq} we aim now to apply the implicit function theorem in order to find solution flows with non-flat free surface. To this
end we need to compute the 
 derivative of $\mathfrak{F}$ with respect to the first variable $\mathcal h$ at the point $(0,\mathcal P_0) \in B\times  \mathcal C (I)$,
\begin{equation}
D_{\mathcal h} \mathfrak{F} (0,\mathcal P_0) \mathcal h =
\lim_{t\to 0} \frac{\mathfrak{F} (t \mathcal h,\mathcal P_0) - \mathfrak{F} (0,\mathcal P_0)}{t}, 
\end{equation}
A plain calculation yields that
\begin{align} \label{D_h_F}
\begin{aligned}
D_{\mathcal h} \mathfrak{F} (0,\mathcal P_0) \, \mathcal h (\theta) &=
- \frac{g R}{P_a} \rho(R,\theta)\, \mathcal h (\theta)
+ \frac{1}{P_a} \Big[ F(R \sin \theta) + R \sin \theta \, \mathcal F (R \sin \theta, \theta)  \Big] \, \mathcal h(\theta) \nonumber \\
&= \frac{\rho(R,\theta)}{P_a} \Big[ -g R + \big( w(R,\theta) + \Omega R \sin \theta \big)^2  \Big] \, \mathcal h(\theta), \nonumber
\end{aligned}
\end{align}
where we employed \eqref{azim_vel} to obtain the last equality. Realistic assumptions on the maximal size of $w$ yield that 
$g R \gg ( w + \Omega R\sin\theta)^2$ for all $\theta \in I$, cf.~Fig.~\ref{fig: fronts}. This allows us to infer that
\begin{equation}
D_{\mathcal h} \mathfrak{F} (0,\mathcal P_0) \colon \mathcal C (I)
\to \mathcal C(I)
\end{equation}
defines a linear topological homeomorphism. By the implicit function theorem in Banach spaces (cf.~\cite{Ber}) there exists for every sufficiently 
small perturbation $\mathcal P$ of $\mathcal P_0$ a unique function $\mathcal h \in \mathcal C (I)$ such that \eqref{P_nondim_formula} 
holds true; this $\mathcal h$ is nontrivial.

\section{Monotonicity relations}\label{monotonicity}
\noindent
While in the previous section we have exploited the non-dimensional relation \eqref{P_nondim_formula} to infer that the free surface in the region of the ACC 
can locally---around the undisturbed state $(\mathcal h \equiv 0,\mathcal P_0)$---be described as a function of the imposed surface pressure, we establish in
this section relations between the monotonicity of the dimensionless surface pressure $\mathcal{P}$ and surface elevation $\mathcal{h}$, 
see Proposition \ref{mon_prop} below, which are in concordance with physical expectations. 
To this end we remark that a bootstrapping procedure \cite{Ber} enables us to transfer smoothness properties of $\mathcal{P}$  to $\mathcal{h}$.
Hence assuming a continuously differentiable pressure $\mathcal P$, which we differentiate now with respect to $\theta$ in \eqref{P_nondim_formula} yields
\begin{align*}
\mathcal P' (\theta) = 
&-\frac{g}{P_a}\int_a^{[1+\mathcal h]R} \rho_{\theta}(r,\theta) \, \dx r
-\frac{g}{P_a}\rho\big( (1+\mathcal h)R, \theta\big) \, \mathcal h' (\theta)R\\
&+\frac{\cot\theta}{P_a}F\big((1+\mathcal h )R\sin\theta\big)+\frac{\mathcal h' (\theta)}{P_a} \, 
\frac{F\big((1+\mathcal h)R\sin\theta\big)}{1+\mathcal h}-\frac{\cot\theta}{P_a}F(a\sin\theta) \\
&+\frac{1}{P_a} \, \mathcal{F}\big((1+\mathcal h)R\sin\theta,\theta\big)\Big((1+\mathcal h)R\cos\theta+\mathcal h'(\theta) \, R \sin\theta\Big) \\
&-\frac{1}{P_a} \, \mathcal{F}(a\sin\theta,\theta) \, a\cos\theta 
+\frac{1}{P_a}\int_a^{[1+\mathcal h]R} [g\rho_{\theta}(r,\theta)+H_r(r,\theta)] \, \dx r \\
&+\frac{1}{P_a} \, \Big( F(a\sin\theta) \cot \theta +\mathcal{F}(a\sin\theta, \theta) \, a\cos\theta +H(a,\theta)\Big),
\end{align*}
thus
\begin{align*}
\mathcal P' (\theta) &= 
\frac{\mathcal h' (\theta)}{P_a} \Bigg[-gR\rho \big((1+\mathcal h)R,\theta\big)+R \, \sin\theta \, \mathcal F \big((1+\mathcal h)R\sin\theta, \theta\big) \\
& \qquad \qquad \;\, +\frac{F\big((1+\mathcal h)R\sin\theta\big)}{1+\mathcal h} \Bigg]
+\frac{1}{P_a}\, H\big((1+\mathcal h )R,\theta\big)	\\
& \quad+\frac{\cot\theta}{P_a}\left[ F\big((1+\mathcal h)R\sin\theta\big)+(1+\mathcal h)R\, \sin\theta \, \mathcal{F} \big((1+\mathcal h)R\sin\theta,\theta\big) \right]; 
\end{align*}
for convenience we wrote simply $\mathcal h$ to mean the evaluation $\mathcal h (\theta)$ in the above calculation.
Making use of the formula for the azimuthal velocity \eqref{azim_vel} we obtain
\begin{equation}\label{rel_ph}
\begin{aligned}
\frac{P_a}{\rho\big((1+\mathcal h)R,\theta\big)} \, \mathcal P'(\theta)&=
\left[-g R+\frac{\Big(w\big((1+\mathcal h)R, \theta\big)+\Omega(1+\mathcal h)R\sin\theta\Big)^2}{1+\mathcal h}\right]\mathcal h' (\theta)\\
&\qquad + \Big(w\big((1+\mathcal h)R, \theta\big)+\Omega(1+\mathcal h)R\sin\theta\Big)^2  \cot\theta  \\
&\qquad+(1+\mathcal h)R \, G\big((1+\mathcal h )R,\theta\big).
\end{aligned}
\end{equation}
It is apparent from the above formula that, in order 
to discuss monotonicity properties of $\mathcal{P}$ and $\mathcal{h}$, we need additional information on the size of the forcing term $G$. Along these lines, we note that, 
according to the discussion in Constantin and Johnson \cite{CJazAcc}, the nonlinear advection terms in \eqref{specialflow} are small perturbations of the linear flow $w_0$ which is governed by 
the system
\begin{equation}\label{linflow}
 \begin{array}{rcl}
  -2\Omega w_0\sin\theta & = & -\frac{1}{\rho}\tilde{p}_r\\
  & &\\
  -2\Omega w_0\cos\theta & = & -\frac{1}{\rho}\cdot\frac{1}{r}\tilde{p}_{\theta}+G(r,\theta),
 \end{array}
\end{equation}
with $\tilde{p}=p+\rho g r-\frac{1}{2}\rho r^2\Omega^2\sin^2\theta$ being the modified pressure absorbing the centripetal and gravitational acceleration.
Using that $\sin\theta\approx\cos\theta$ in the region of the ACC (recall that we restrict $\theta$ to the interval $I=\big[\frac{3\pi}{4}-\frac{\pi}{18},\frac{3\pi}{4}+\frac{\pi}{18}\big]$, cf.~\eqref{latitude_ACC}), and that the meridional pressure gradient is relatively small compared to the radial gradient, we have from
\eqref{linflow} that, at the leading order balance, 
\begin{equation}\label{forcing}
 G(r,\theta)\approx -2\Omega w_0\cos\theta.
\end{equation}
We are now able to formulate the following monotonicity result.
\begin{prop}\label{mon_prop}
Let $\theta\in I$ and assume that the forcing term $G$ satisfies \eqref{forcing}. 
Then the following assertions hold true.
\begin{enumerate}[(i)]
\item  
If $\mathcal h' (\theta)\geq 0$ then $\mathcal P' (\theta)<0$.  \label{prop_mon_i}
\item  
If  $\mathcal P' (\theta) \geq 0$  then $\mathcal h' (\theta)< 0$. \label{prop_mon_ii}
\end{enumerate}
\end{prop}

\begin{proof}
 Let us first consider the sum
\begin{equation*}
\Big(w((1+\mathcal{h})R, \theta)+\Omega(1+\mathcal{h})R\sin\theta\Big)^2 \cot\theta +(1+\mathcal{h})R \, G\big((1+\mathcal{h})R,\theta\big)
\end{equation*}
 and expand it as
 \begin{equation}\label{second_third}
 \begin{aligned}
&\Big(w((1+\mathcal{h})R, \theta)^2 +2\Omega R (1+\mathcal{h})w((1+\mathcal{h})R, \theta)\sin\theta\Big)\cot\theta\\
&+\Omega R(1+\mathcal{h})\left[\Omega R(1+\mathcal{h})\sin\theta-2w_0((1+\mathcal{h})R, \theta)\right]\cos\theta .
 \end{aligned}
 \end{equation}
Given that the maximal speed of the ACC does not exceed $3$ m s$^{-1}$ (cf. \cite{CJazAcc}), clearly 
$\Omega R(1+\mathcal{h})\sin\theta$ is by far greater than $2w_0$. Therefore, the expression in \eqref{second_third} is clearly negative for all
$\theta\in [\frac{3\pi}{4}-\frac{\pi}{18},\frac{3\pi}{4}+\frac{\pi}{18}]$.
Also, appealing to the same arguments pertaining to the sizes of $w_0, \Omega$ and $R$ we can conclude that 
\begin{equation*}
-g R+\frac{\Big(w((1+\mathcal{h})R, \theta)+\Omega(1+\mathcal{h})R\sin\theta\Big)^2}{1+\mathcal{h}} < 0
\end{equation*} 
within $I$. Thus, assuming 
$\mathcal{h}^{\prime}(\theta)\geq 0$, we see from \eqref{rel_ph} that $\mathcal{P}^{\prime}(\theta)<0$, which shows \eqref{prop_mon_i}.
The proof of \eqref{prop_mon_ii} follows by the same arguments.
\end{proof}

\vspace{1em}
\noindent
\textbf{Acknowledgments}\\
\noindent
C.~I.~Martin would like to acknowledge the support of the Austrian Science Fund (FWF) under research grant P 30878-N32. R.~Quirchmayr acknowledges the support of FWF under research grant J 4339-N32.


\begin{thebibliography}{ll}
\bibitem{BasDCDS} B. Basu, \textit{On an exact solution of a nonlinear three-dimensional model in ocean flows with equatorial undercurrent and linear variation in density}.
 Discrete Contin. Dyn. Syst. A. 39 (2019) no. 8, 4783-4796
%
\bibitem{Ber} M. S. Berger, Nonlinearity and Functional Analysis, Academic Press, New York, 1977.
%
\bibitem{CWW} R. M. Chen, S. Walsh and M.  Wheeler, \textit{Existence and qualitative theory for stratified solitary water waves}, Ann. Inst. H. Poincaré Anal. Non Linéaire 35 (2018) no. 2, 517--576.
%
\bibitem{ChuEsc} J. Chu and J. Escher, \textit{Steady periodic equatorial water waves with vorticity}, Discrete Contin. Dyn. Syst. A. 39 (2019) no. 8, 4713-4729. 
\bibitem{ChuIonYang} J. Chu, D. Ionescu-Kruse and Y. Yang, \textit{Exact solution and instability for geophysical waves at arbitrary latitude}, 
Discrete Contin. Dyn. Syst. A. 39 (2019) no. 8,4399-4414 
%
\bibitem{CGeoResLett} A. Constantin, \textit{On the modelling of equatorial waves}, Geophys. Res. Lett. 39 (2012), L05602.
%
\bibitem{CoGeoPhys} A. Constantin, \textit{An exact solution for equatorially trapped waves}, J. Geophys. Res. Oceans 117, (2012), C05029.
%
\bibitem{CoPhysOc13} A. Constantin, \textit{Some three-dimensional nonlinear equatorial flows}, J. Phys. Oceanogr. 43 (2013), 165--175.
%
\bibitem{CoPhysOc14} A. Constantin, \textit{Some nonlinear, equatorially trapped, nonhydrostatic internal geophysical waves}, J. Phys. Oceanogr. 44 (2014), no. 2, 781--789.
%
\bibitem{CI} A. Constantin and R. I. Ivanov,  \textit{A Hamiltonian approach to wave-current interactions in two-layer fluids}, Physics of Fluids 27 (2015), 086603.
%
\bibitem{CJ} A. Constantin and R. S. Johnson, \textit{The dynamics of waves interacting with the Equatorial Undercurrent}, Geophysical and Astrophysical Fluid Dynamics, 109
(2015), no. 4, 311--358.
%
\bibitem{cim} A.~Constantin, R. I. Ivanov and C. I. Martin,  \textit{Hamiltonian formulation for wave-current interactions in stratified rotational flows}, Arch. Ration. Mech. Anal 221 (2016), 1417--1447.
%
\bibitem{CJaz} A. Constantin and R. S. Johnson, \textit{An exact, steady, purely azimuthal equatorial flow with a free surface}, J. Phys. Oceanogr. 46 (2016), no. 6, 1935-1945,
%
\bibitem{CJazAcc}  A. Constantin and R. S. Johnson, \textit{An exact, steady, purely azimuthal flow as a model for the Antarctic Circumpolar Current}, 
J. Phys. Oceanogr., 46 (2016), no. 12, 3585-3594,
%
\bibitem{CJPoF} A. Constantin and R. S. Johnson,  \textit{A nonlinear, three-dimensional model for ocean flows, motivated by some observations of the Pacific Equatorial Undercurrent and thermocline}, Physics of Fluids 29 
(2017), 056604, 
%
\bibitem{CJJPhysOc19} A. Constantin and R. S. Johnson,  \textit{On the nonlinear, three-dimensional structure of equatorial oceanic flows}. J. Phys. Oceanogr. 2019, https://doi.org/10.1175/JPO-D-19-0079.1
%
\bibitem{CICpam}  A. Constantin and R. I. Ivanov,  \textit{Equatorial wave-current interactions}. Comm. Pure Appl. Math. to appear 2019.
%
\bibitem{Danabas} G. Danabasoglu, J. C. McWilliams and P. R. Gent, \textit{The role of mesoscale tracer transport in the global ocean circulation}. Science 264 (264) 1994, 1123-1126.
%
\bibitem{dtwcc} K. A. Donohue, K. L. Tracey, D. R. Watts, M. P. Chidichimo and T. K. Chereskin.  \textit{Mean Antarctic Circumpolar Current transport measured in Drake Passage.}
Geophysical Research Letters 43 (2016), no. 22, 760-767.
%
\bibitem{EMM} J. Escher, A.-V. Matioc and B.-V. Matioc, \textit{On stratified steady periodic water waves with linear density distribution and stagnation points}, J. Diff. Equations  251 (2011), 2932--2949.
%
\bibitem{FedBr} A. V. Fedorov and J. N. Brown, \textit{Equatorial waves}. In \textit{Encyclopedia of ocean sciences}, edited by J. Steele, 
(Academic Press: New York, 2009), 3679--3695.

\bibitem{Fir} Y. L. Firing, T. K. Chereskin and M. R. Mazloff. \textit{Vertical structure and transport of the Antarctic Circumpolar Current in Drake Passage from direct velocity 
observations.} J. Geophys. Res. 116 (2011), C08015.
%
\bibitem{GeyQui} A. Geyer and R. Quirchmayr. \textit{Shallow water models for stratified equatorial flows}. Discrete Contin. Dyn. Syst. A. 39 (2019) no. 8, 4533-4545.
%
\bibitem{HazMary} S. V. Haziot and K. Marynets. \textit{Applying the stereographic projection to modeling of the flow of the antarctic circumpolar current.} Oceanography 31 (2018), no. 3, 68-75.
%
\bibitem{HazDCDS} S. V. Haziot. \textit{Study of an elliptic partial differential equation modelling the Antarctic circumpolar current}. Discrete and Continuous Dynamical Systems- Series A 39 (2019), no. 8, 4415-4427.
%
\bibitem{HenEjmb} D. Henry. An exact solution for equatorial geophysical water waves with an underlying current, \emph{Eur. J. Mech. B Fluids} 38 (2013), 18--21.
%
\bibitem{HM1} D. Henry and B.-V. Matioc, \textit{On the existence of steady periodic capillary-gravity stratified water waves}, {Ann. Sc. Norm. Super. Pisa Cl. Sci.} (5) 12 (2013), no. 4, 955--974.
%
\bibitem{HM2} D. Henry and A.-V. Matioc, \textit{Global bifurcation of capillary--gravity-stratified water waves}, { Proc. Roy. Soc. Edinburgh Sect. A} {144} (2014), no. 4, 775--786.     
%
\bibitem{HenMarJDE} D. Henry and C.-I. Martin. \textit{Free-surface, purely azimuthal equatorial flows in spherical coordinates with stratification} J. Differential Equations
266 (2019), no. 10, 6788-6808.
%
\bibitem{HenMarARMA} D. Henry and C. I. Martin. \emph{Azimuthal equatorial flows with variable density in spherical coordinates}, Arch. Ration. Mech. Anal. 233 (2019), 497-512.
%
\bibitem{MarshForcing} E. Howard, A. M. Hogg, S. Waterman, and D. P. Marshall. \textit{The injection of zonal momentum by buoyancy forcing in a Southern Ocean model.}  J. Phys. Oceanogr. 45 (2015), 259-271
%
\bibitem{HsuMarNA} H.-C. Hsu and C. I. Martin. \textit{Free-surface capillary-gravity azimuthal equatorial flows}, Nonlinear Analysis: Theory, Methods and Applications 144 (2016), 1-9.
%
\bibitem{HsuMarAcc} H.-C. Hsu and C.I. Martin. \textit{On the existence of solutions and the pressure function related to the Antarctic Circumpolar Current,} Nonlinear Analysis: Theory, Methods and Applications 155 (2017), 285-293.
%
\bibitem{Ion} D. Ionescu-Kruse. \textit{A three-dimensional autonomous nonlinear dynamical system modelling equatorial ocean flows}, J. Differential Equations 264 (2018) no. 7, 4650-4668.
%
\bibitem{Ivch} V. O. Ivchenko and K. J. Richards. \textit{The dynamics of the Antarctic Circumpolar Current}, J. Phys. Oceanogr. 26 (1996), 753-774.
%
\bibitem{KesMcP} W. S. Kessler and M. J. McPhaden, \textit{Oceanic equatorial waves and the 1991-93 El Ni\~{n}o}, J. Climate 8 (1995), 1757-1774.
%
\bibitem{Kob} H. Kobayashi, A. Abe-Ouchi and A. Oka. \emph{Role of Southern Ocean stratification in glacial atmospheric CO$_2$ reduction evaluated by a three-dimensional ocean general circulation 
model}. Paleooceanography 30 (2015), no. 9, 1202-1216.
%
\bibitem{Marsh16} D. P. Marshall, D. R. Munday, L. C. Allison, R. J. Hay, and H. L. Johnson. Gill’s model of the Antarctic Circumpolar Current, revisited: The role of latitudinal variations 
in wind stress. \emph{Ocean Modelling} 97 (2016), 37-51
%
\bibitem{MarNonl} C. I. Martin. \textit{Constant vorticity water flows with full Coriolis term}, Nonlinearity 32 (2019) no. 7, 2327-2336.
%
\bibitem{AMJPhA} A.-V. Matioc. \textit{An exact solution for geophysical equatorial edge waves over a sloping beach}, J. Phys. A 45 (2012), no. 36, 365501, 10 pp. 
%
\bibitem{AMApplAna} A.-V. Matioc. \textit{Exact geophysical waves in stratified fluids}, Appl. Anal. 92 (2013), no. 11, 2254-2261.
%
\bibitem{MatMatJnmp} A.-V. Matioc and B.-V. Matioc. \textit{On periodic water waves with Coriolis effects and isobaric streamlines}, J. Nonl. Math. Phys. 19 (2012), suppl. 1, 1240009.
%
\bibitem{Mary}  K. Marynets. \textit{The Antarctic Circumpolar Current as a shallow-water asymptotic solution of Euler's equation in spherical coordinates}. 
Deep-Sea Research Part II: Topical Studies in Oceanography 160 (2019), 58-62.
%
\bibitem{Mas} S. A. Maslowe, \textit{Critical layers in shear flows}, Ann. Rev. Fluid. Mech., 18 (1986), 405-432.
%
\bibitem{McC} J. P. McCreary, \textit{Modeling equatorial ocean circulation}, Ann. Rev. Fluid Mech. 17 (1985), 359-409.
%
\bibitem{Olb} D. Olbers, D. Borowski, C. V\"olker and J.-O. W\"olff. \textit{The dynamical balance, transport and circulation of the Antarctic Circumpolar Current}, Antarctic. Sci
16 (2004), 439-470.
%
\bibitem{Phillips} H. Phillips, B. Legresy and N. Bindoff. \textit{Explainer: how the Antarctic Circumpolar Current helps keep Antarctica frozen}. The Conversation, November 15, 2018.
%
\bibitem{Quir} R. Quirchmayr. \textit{A steady, purely azimuthal flow model for the Antarctic Circumpolar Current}. Monatshefte f\"ur Mathematik 187 (2018), no. 3, 565-572.
%
\bibitem{Rin} S. R. Rintoul, C. Hughes and D. Olbers, The Antarctic Circumpolar Current system. Ocean Circulation and Climate: \textit{Observing and Modelling the Global Ocean}, G. Seidler,
J. Church and J. Gould, Eds., Academic Press, San Diego, 271-302.
%
\bibitem{Wal1} S. Walsh, \textit{Stratified steady periodic water waves}, SIAM J. Math. Anal. 41 (2009), 1054--1105.
%
\bibitem{Whe} M. H. Wheeler. \textit{On stratified water waves with critical layers and Coriolis forces}, Discrete Contin. Dyn. Syst. A. 39 (2019) no. 8, 4747-4770. 
\end{thebibliography}
\end{document}